\documentclass[10pt,conference]{IEEEtran}
\usepackage{graphicx, amsmath, amssymb}
\begin{document}


%

\title{Degrees of Freedom of Multi-Source Relay Networks}

\author{
\authorblockN{Sang-Woon Jeon, Sae-Young Chung}
\authorblockA{School of EECS\\ KAIST\\ Daejeon, Korea \\
Email: swjeon@kaist.ac.kr, sychung@ee.kaist.ac.kr} \and
\authorblockN{Syed A. Jafar}
\authorblockA{Electrical Engineering and Computer Science\\
University of California, Irvine\\
Irvine, CA 92697, USA\\
Email: syed@uci.edu } }

\maketitle


\newtheorem{definition}{Definition}
\newtheorem{theorem}{Theorem}
\newtheorem{lemma}{Lemma}
\newtheorem{example}{Example}
\newtheorem{corollary}{Corollary}
\newtheorem{proposition}{Proposition}
\newtheorem{conjecture}{Conjecture}
\newtheorem{remark}{Remark}

\def \diag{\operatornamewithlimits{diag}}
\def \min{\operatornamewithlimits{min}}
\def \max{\operatornamewithlimits{max}}
\def \log{\operatorname{log}}
\def \max{\operatorname{max}}
\def \rank{\operatorname{rank}}
\def \out{\operatorname{out}}
\def \exp{\operatorname{exp}}
\def \arg{\operatorname{arg}}
\def \E{\operatorname{E}}
\def \tr{\operatorname{tr}}
\def \SNR{\operatorname{SNR}}
\def \SINR{\operatorname{SINR}}
\def \dB{\operatorname{dB}}
\def \ln{\operatorname{ln}}
\def \th{\operatorname{th}}

\begin{abstract}
We study a multi-source Gaussian relay network consisting of $K$ source--destination pairs having $K$ unicast sessions.
We assume $M$ layers of relays between the sources and the destinations.
We find achievable degrees of freedom of the network.
Our schemes are based on interference alignment at the transmitters and symbol extension and opportunistic interference cancellation at the relays.
For $K$-$L$-$K$ networks, i.e., $2$-hop network with $L$ relays, we show $\min\{K,K/2+L/(2(K-1))\}$ degrees of freedom are achievable.
For $K$-hop networks with $K$ relays in each layer, we show the full $K$ degrees of freedom are achievable provided
that $K$ is even and the channel distribution satisfies a certain symmetry.
\end{abstract}

\section{Introduction} \label{sec:intro}
Capacity characterization of multi-source networks is one of the fundamental problems in network information theory.
However the capacity is not fully characterized even for the simplest setting of the two-user interference channel \cite{Han:81} that leads to surging interests and demands on approximate capacity characterization.
Recently there has been a series of significant progress on approximate capacity characterization \cite{Etkin:08, Bresler:07, Nam:08, Viveck1:08, Tiangao:09, Viveck3:09, Viveck2:09, Viveck1:09}.
The capacity region of the two-user Gaussian interference channel was characterized by Etkin, Tse, and Wang within one bit precision \cite{Etkin:08} and the sum capacity $C_{\Sigma}(P)$ of the $K$-user time-varying Gaussian interference channel was characterized as
\begin{equation*}
C_{\Sigma}(P)=\frac{K}{2}\log(P)+o(\log(P))
\end{equation*}
by Cadambe and Jafar \cite{Viveck1:08}, where $P$ denotes the signal to noise ratio (SNR). That is, the degrees of freedom (DoF) or capacity pre-log term of the $K$-user interference channel is given by $K/2$ \footnote{Unless otherwise stated, we assume time-varying channel in the rest of the paper}.
To achieve $K/2$ DoF, a new interference management technique called the interference alignment was used, which minimizes the dimensions occupied by interference at destinations by aligning the interference from multiple unintended sources.
The interference alignment can also be used to compute achievable DoF of some other channels such as the $K$-user multiple-input multiple-output (MIMO) channel \cite{Tiangao:09}, the deterministic $K$-user interference channel \cite{Viveck3:09}, and the $X$-network \cite{Viveck2:09} in which each of $S$ sources has messages for $D$ destinations, i.e., total of $SD$ message sets.

In this paper, we consider relay networks consisting of multiple sources, multiple corresponding destinations, and multiple relays. Relays have been traditionally used for extending coverage in wireless environments, e.g., amplify-and-forward (AF) based relays. Although the DoF is upper bounded by $\frac{K}{2}$ for fully connected $K$-user interference channels~\cite{HostMadsen:05}, with help of relays it may be possible to improve the DoF also.

The work \cite{Viveck2:09} has applied their $X$ network results to a two-hop network with $S$ sources $D$ destinations with $L$ relays between them.
Assuming $K=S=D$, they showed that $\frac{KL}{K+L-1}$ DoF is achievable\footnote{In~\cite{Viveck2:09} it was half of this. Since we assume full duplex relaying in this paper, we have adjusted it accordingly.}.
Notice that whereas a trivial upper bound assuming perfect cooperation between relays is $K$ if $L\geq K$, the achievable DoF of $\frac{KL}{K+L-1}$ converges to $K$ only if $L\to\infty$.
Hence one of the basic questions about relay networks is the minimum number of required relays for achieving the optimal $K$ DoF.
The works \cite{Bolcskei:06, Morgenshtern:07} have addressed similar questions and shown that with $K$ fixed the sum rate of $K\log(L)+O(1)$ is achievable if $L\to\infty$ \cite{Bolcskei:06}.\footnote{Some assumptions such as the availability of channel state information are different from ours.}

The main contributions of this paper are the follows.
\begin{itemize}
\item For $K$-$L$-$K$ networks, i.e., $2$-hop network with $L$ relays, we show that $\min\big\{K,\frac{K}{2}+\frac{L}{2(K-1)}\big\}$ DoF is achievable.
Hence the optimal $K$ DoF is achievable if $L\geq K(K-1)$.
To show the achievability, interference alignment combined with distributed interference cancellation using multiple relays is used over multiple symbols (symbol extension) to utilize more diversity provided by time-varying channels.
A similar interference cancellation technique called interference neutralization was used for deterministic and non-fading Gaussian $2$-user $2$-hop networks~\cite{Mohajer:08, Mohajer:09}, where multiple relays are cooperatively used to cancel interference. In our case, such distributed interference cancellation is combined with symbol extension in a more general network.

\item For $K$-hop networks with $K$ relays in each layer, we show that the optimal $K$ DoF is achievable if $K$ is even and the probability of channel matrix is a function of its Frobenius norm only.
We apply a new technique called opportunistic interference cancellation where the relays in each layer delay-amplify-and-forward their received signal vector by waiting for an appropriate channel instance in the next hop such that the overall channel matrix from the sources to destinations become a scaled identity matrix.
Hence $K$ source--destination (S--D) pairs can communicate concurrently without interference.
This is related to the opportunistic interference cancellation for finite-field networks~\cite{Jeon:09}, but our scheme for Gaussian relay networks in this paper works differently.
For both cases, opportunistic interference cancellation is applied and the optimal DoF of Gaussian networks and the optimal sum rate of finite-field networks are achieved in certain cases.
Notice that $K(K-1)$ relays are again required to obtain $K$ DoF.
\end{itemize}

The concept of opportunistic channel pairing can be found in \cite{JeonITA:09, JeonISIT:09, Nazer:09} for the finite-field interference channel and in \cite{Nazer:09, Jafar:09} for the Gaussian interference channel, which it was called ergodic interference alignment.
For the multi-hop case, opportunistic interference cancellation was applied for finite-field networks in \cite{JeonISIT:09}.

\section{System Model} \label{sec:sys_model}
Throughout the paper, we use notations $\mathbf{A}$ and $\mathbf{a}$ to denote a matrix and a vector, respectively.
The transpose, conjugate transpose, and Frobenius norm of $\mathbf{A}$ (or $\mathbf{a}$) are denoted by $\mathbf{A}^T$, $\mathbf{A}^{\dagger}$, and $\|\mathbf{A}\|_F$ (or $\mathbf{a}^T$, $\mathbf{a}^{\dagger}$, and $\|\mathbf{a}\|_F$), respectively.
The diagonal matrix having $a_i$ as the $i$-th diagonal element, the $n_1\times n_1$ identity matrix, and the $n_1\times n_2$ all-zero matrix are denoted by $\diag(a_1,\cdots, a_{n_1})$, $\mathbf{I}_{n_1}$, and $\mathbf{0}_{n_1\times n_2}$, respectively.
We also use $\bar{\mathbf{A}}$ and $\bar{\mathbf{a}}$ to denote $N$-symbol-extended matrix and vector consisting of $a[t=1]$ through $a[t=N]$, respectively.
That is, $\bar{\mathbf{A}}=\diag(a[1],\cdots,a[N])$ and $\bar{\mathbf{a}}=[a[1],\cdots,a[N]]^T$.

\subsection{Gaussian Relay Networks}
We study a $M$-hop relay network having $M+1$ layers with $K_m$ nodes in the $m$-th layer.
The nodes in the first and the last layer are the sources and the destinations, respectively.
Thus $K=K_1=K_{M+1}$ is the number of S--D pairs.
For simplicity, let us denote the $i$-th node in the $m$-th layer by node $(i,m)$.

For $m\in\{1,\cdots,M\}$, let $x_{i,m}[t]$ denote the transmit signal of node $(i,m)$ at time $t$.
Then the received signal $y_{j,m}[t]$ of node $(j,m+1)$ at time $t$ is given by\footnote{The subscript $m$ in $y_{j,m}[t]$ means the $m$-th hop.}
\begin{equation*}
y_{j,m}[t]=\sum_{i=1}^{K_m}h_{ji,m}[t]x_{i,m}[t]+z_{j,m}[t],
\end{equation*}
where $h_{ji,m}[t]$ is the channel from node $(i,m)$ to node $(j,m+1)$ at time $t$ and $z_{j,m}[t]$ is the noise at node $(j,m+1)$ at time $t$.
The noise terms $z_{j,m}[t]$'s are independent and identically distributed (i.i.d.) complex Gaussian with zero-mean and unit-variance.
We assume time-varying channels such that $h_{ji,m}[t]$'s are i.i.d. drawn from a continuous distribution and $\Pr(h_{ji,m}[t]=h)=0$ for all $h\in\mathbb{C}$.
We assume every source and relay has the same power constraint $P$.

Let us denote the transmitted and received signal vectors of the $m$-th hop by $\mathbf{x}_m[t]=[x_{1,m}[t], \cdots, x_{K_m,m}[t]]^T$ and $\mathbf{y}_m[t]=[y_{1,m}[t], \cdots, y_{K_{m+1},m}[t]]^T$, respectively.
Then the input output relation of the $m$-th hop can be represented as
\begin{equation}
\mathbf{y}_m[t]=\mathbf{H}_m[t]\mathbf{x}_m[t]+\mathbf{z}_m[t],
\label{eq:input_output_vec}
\end{equation}
where $\mathbf{H}_m[t]$ is the channel matrix of the $m$-th hop whose $(j,i)$-th element is given by $h_{ji,m}[t]$ and $\mathbf{z}_m[t]=[z_{1,m}[t],\cdots,z_{K_{m+1},m}[t]]^T$ is the noise vector of the $m$-th hop.
The channel state information (CSI) is assumed to be available at all nodes, i.e., each node knows $\mathbf{H}_1[t]$ to $\mathbf{H}_M[t]$ at time $t$.
For simplicity, we omit time index $t$ in the rest of the paper.

\subsection{Degrees of Freedom}
The $i$-th source sends a message $W_i \in \{1,2,\ldots,2^{nR_i(P)}\}$ to its destination at a rate of $R_i(P)$ during $n$ channel uses.
The rate tuple $(R_1(P),\cdots,R_K(P))$ is said to be achievable if the probability of error for all S--D pairs can be made arbitrarily small by choosing large enough $n$.
The capacity region $\mathcal{C}(P)$ is the convex hull of the closure of all achievable rate tuples and the sum capacity $C_{\Sigma}(P)$ is the supremum of all achievable sum rates.
Then the DoF is defined as
\begin{equation*}
d_{\Sigma} \triangleq \lim_{P\to\infty}\frac{C_{\Sigma}(P)}{\log P}.
\end{equation*}

\section{Achievability for $K$-$L$-$K$ Networks}
In this section, we consider a two-hop network with $K_1=K_3=K$ and $K_2=L$ and assume $L\geq K$, which is denoted by the $K$-$L$-$K$ network.

\subsection{Interference Cancellation and Alignment}
We assume AF relaying.
Because multiple replicas of a transmit signal interfere with unintended destinations through multiple relays, we can make the replicas cancel each other by appropriately choosing the gains at relays. Since we assume symbol extension, i.e., we send messages using multiple channel instances, we have vectorized AF, i.e., each relay can multiply its received vector by a matrix.

\subsubsection{AF Based Relay}
Consider $N$ symbol extension.
Then the received signal vector of the $j$-th relay can be represented as
\begin{equation}
\bar{\mathbf{y}}_{j,1}=\sum_{i=1}^K\bar{\mathbf{H}}_{ji,1}\bar{\mathbf{x}}_{i,1}+\bar{\mathbf{z}}_{j,1},
\label{eq:input_output_1hop}
\end{equation}
where $j\in\{1,\cdots,L\}$ and
\begin{eqnarray*}
\bar{\mathbf{x}}_{i,1}\!\!\!\!\!\!\!\!\!\!\!&&=[x_{i,1}[1],\cdots,x_{i,1}[N]]^T\nonumber\\
\bar{\mathbf{y}}_{j,1}\!\!\!\!\!\!\!\!\!\!\!&&=[y_{j,1}[1],\cdots,y_{j,1}[N]]^T\nonumber\\
\bar{\mathbf{z}}_{j,1}\!\!\!\!\!\!\!\!\!\!\!&&=[z_{j,1}[1],\cdots,z_{j,1}[N]]^T\nonumber\\
\bar{\mathbf{H}}_{ji,1}\!\!\!\!\!\!\!\!\!\!\!&&=\diag(h_{ji,1}[1],\cdots,h_{ji,1}[N]).
\end{eqnarray*}
Each relay transmits $N$ linear combinations of its $N$ received signals to the destinations.
Specifically, the transmit signal vector of the $j$-th relay is given by
\begin{equation}
\bar{\mathbf{x}}_{j,2}=\mathbf{\Gamma}_j \bar{\mathbf{y}}_{j,1},
\label{eq:relay_func}
\end{equation}
where
\begin{equation*}
\mathbf{\Gamma}_j=\left[
                    \begin{array}{ccc}
                      \gamma_j[1,1] & \cdots & \gamma_j[1,N] \\
                      \vdots & \ddots & \vdots \\
                      \gamma_j[N,1] & \cdots & \gamma_j[N,N] \\
                    \end{array}
                  \right]
\end{equation*}
represents the gain matrix of the $j$-th relay.\footnote{We assume block Markov coding is used at relays, i.e., each relay collects $N$ symbols, applies a linear transform, and sends it in the next $N$ time slots. Therefore, there will be one block delay at the relays. To simplify notations, we omit block indices.}
Similarly, the received signal vector of the $k$-th destination can be represented as
\begin{equation}
\bar{\mathbf{y}}_{k,2}=\sum_{j=1}^L\bar{\mathbf{H}}_{kj,2}\bar{\mathbf{x}}_{j,2}+\bar{\mathbf{z}}_{k,2},
\label{eq:input_output_2hop}
\end{equation}
where $k\in\{1,\cdots,K\}$.
Combining (\ref{eq:input_output_1hop}) through (\ref{eq:input_output_2hop}), we obtain
\begin{eqnarray} \label{eq:input_output}
\bar{\mathbf{y}}_{k,2}\!\!\!\!\!\!\!\!\!\!&&=\sum_{j=1}^{L}\bar{\mathbf{H}}_{kj,2}\mathbf{\Gamma}_j\bar{\mathbf{H}}_{jk,1}\bar{\mathbf{x}}_{k,1}+\sum_{i=1,i\neq k}^K \sum_{j=1}^{L}\bar{\mathbf{H}}_{kj,2}\mathbf{\Gamma}_j\bar{\mathbf{H}}_{ji,1}\bar{\mathbf{x}}_{i,1}\nonumber\\
&&+\sum_{j=1}^L\bar{\mathbf{H}}_{kj,2}\mathbf{\Gamma}_j\bar{\mathbf{z}}_{j,1}+\bar{\mathbf{z}}_{k,2}.
\end{eqnarray}
Notice that the first term is the intended signals and the second term is the interfering signals and the third term is the noise propagation due to AF based relaying.
Note that the following condition guarantees that the interference from the $i$-th source to the $k$-th destination ($i\neq k$) will be nullified.
\begin{equation}
\sum_{j=1}^{L}\bar{\mathbf{H}}_{kj,2}\mathbf{\Gamma}_j\bar{\mathbf{H}}_{ji,1}\bar{\mathbf{v}}_{i}=\mathbf{0}_{N\times 1}.
\end{equation}
Fig. \ref{fig:int_mit} illustrates this.
The detailed analysis will be given in Lemma \ref{lemma:n_cancel}.

\begin{figure}[t!]
  \begin{center}
  \scalebox{0.45}{\includegraphics{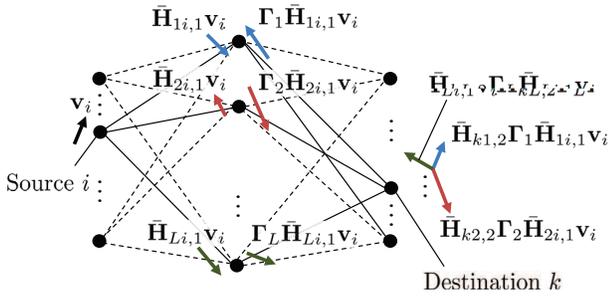}}
  \caption{Interference cancellation using relays.}
  \label{fig:int_mit}
  \end{center}
\end{figure}

\subsubsection{Transmission scheme}
For transmission, we only use the time slots where channel gains satisfy $g_{\min}\leq |\mathbf{h}_{ji,m}[t]|\leq g_{\max}$ for all $i$ and $j$, where $g_{\min}>0$, $g_{\max}>0$, and $g_{\max}>g_{\min}$.
Since we assume $\Pr(\mathbf{h}_{ji,m}[t]=h)=0$, the probability of slot utilization can be arbitrarily close to one by setting $g_{\min}$ and $g_{\max}$ as arbitrarily small and large, respectively.
As a result introducing $g_{\min}$ and $g_{\max}$ does not affect the DoF.

For $N$ symbol extension, we allocate $N_1+N_2$ symbols to the first S--D pair and $N_1+N_3$ symbols to each of the remaining $K-1$ S--D pairs.
The first source transmits $N_1$ symbols without transmit beamforming and transmits $N_2$ symbols via beamforming vectors $\mathbf{v}^{(1)}_1$ to $\mathbf{v}^{(N_2)}_1$.
Similarly, for $i\in\{2,\cdots, K\}$, the $i$-th source transmits $N_1$ symbols without beamforming and $N_3$ symbols via $\mathbf{v}^{(1)}_i$ to $\mathbf{v}^{(N_3)}_i$.
The interference caused by the $N_1$ symbols will be cancelled by using relay coefficients and the interference caused by the remaining symbols will be aligned by using transmit beamforming.

\subsection{DoF of $K$-$L$-$K$ networks}
The following lemma shows that each source can transmit $N_1$ symbols without interfering with unintended destinations if we set $N_1\leq \min\left\{\left\lfloor\frac{LN^2-1}{K(K-1)N}\right\rfloor,N\right\}$.
\begin{lemma} \label{lemma:n_cancel}
Suppose a $K$-$L$-$K$ network with $N$ time extension.
Then there exist $\mathbf{\Gamma}_i$'s such that each source transmits up to $\min\left\{\left\lfloor\frac{LN^2-1}{K(K-1)N}\right\rfloor,N\right\}$ symbols without interfering with unintended destinations.
\end{lemma}
\begin{proof}
we refer readers to the full paper \cite{Jeon2:09}.
\end{proof}

Then we apply the interference alignment technique to the remaining interference and obtain the achievable DoF in the following theorem.
\begin{theorem} \label{th:K_L_K}
Suppose a $K$-$L$-$K$ network. Then $d_{\Sigma}\geq \min\left\{K,\frac{K}{2}+\frac{L}{2(K-1)}\right\}$.
\end{theorem}
\begin{proof}
First of all, we briefly discuss the power constraint issue.
Since each channel used for transmission satisfies $g_{\min}\leq|h_{ji,m}|\leq g_{\max}$ and the absolute value of each relay coefficient can also be bounded between some minimum and maximum values, the noise term in (\ref{eq:input_output}) does not affect the DoF.
For detailed proof, we refer readers to the full version of this paper.

For $L>K(K-1)$ we set $N_1=1$, $N_2=N_3=0$, and $N=1$.
Then, from the result of Lemma \ref{lemma:n_cancel},
\begin{equation}
\min\left\{\left\lfloor\frac{LN^2-1}{K(K-1)N}\right\rfloor,N\right\}=1
\end{equation}
symbol can be cancelled at each unintended destination.
Thus $K$ DoF is achievable without symbol extension.

For $L=K(K-1)$, we set $N_1=n$, $N_2+N_3=0$, and $N=n+1$, where $n>0$ is an arbitrary integer.
Then
\begin{equation*}
\min\left\{\left\lfloor\frac{LN^2-1}{K(K-1)N}\right\rfloor,N\right\}\geq n
\end{equation*}
symbols can be cancelled.
Thus the achievable DoF is given by
\begin{equation*}
\sup_nK\frac{n}{n+1}=K.
\end{equation*}

For $K\leq L< K(K-1)$, we set $N_1=\left\lfloor\frac{(N_2+N_3)L-1}{K(K-1)-L}\right\rfloor$, $N_2=(n+1)^T$, and $N_3=n^T$, and $N=\left\lceil\frac{(N_2+N_3)L-1}{K(K-1)-L}\right\rceil+N_2+N_3$, with $T=(K-1)(K-2)-1$, where $n>0$ is an arbitrary integer.
Then $N_1$ symbols can be cancelled at each unintended destination because
\begin{equation*}
\min\left\{\left\lfloor\frac{LN^2-1}{K(K-1)N}\right\rfloor,N\right\}\geq N_1,
\end{equation*}
where we use the fact that $N\geq \frac{(N_2+N_3)L-1}{K(K-1)-L}+N_2+N_3$.

Notice that since the relays transmit linear combinations of their received signals, the $N$-symbol-extended channel matrix from the $i$-th source to the $k$-th destination is given by
\begin{equation*}
\bar{\mathbf{G}}_{ki}=\sum_{j=1}^{L}\bar{\mathbf{H}}_{kj,2}\mathbf{\Gamma}_j\bar{\mathbf{H}}_{ji,1},
\end{equation*}
which means one can regard the resulting network as the $K$-user interference channel having the $N$-symbol-extended channel matrix $\bar{\mathbf{G}}_{ki}$.
Hence the remaining interference from $K-1$ unintended sources can be aligned at each destination.
We apply the interference alignment technique in \cite{Viveck1:08} to align the remaining interference and refer Appendix III in \cite{Viveck1:08} for the detailed proof.
Let $d_i(n)$ be the number of transmit symbols of the $i$-th source divided by $N$.
Then,
\begin{eqnarray}
\!\!\!\!\!\!\!\!\!d_1(n)\!\!\!\!\!\!\!\!\!\!&&= \frac{N_1+N_2}{N}\nonumber\\
\!\!\!\!\!\!\!\!\!&&\geq\frac{\frac{(N_2+N_3)L-1}{K(K-1)-L}+N_2-1}{\frac{(N_2+N_3)L-1}{K(K-1)-L}+N_2+N_3+1}\nonumber\\
\!\!\!\!\!\!\!\!\!&&=\frac{K(K-1)N_2+LN_3-(K(K-1)+L+1)}{K(K-1)(N_2+N_3)+(K(K-1)-L-1)}.
\label{eq:d_1}
\end{eqnarray}
Similarly, we obtain
\begin{eqnarray}
\!\!\!\!\!\!\!\!\!d_i(n)\!\!\!\!\!\!\!\!\!\!&&= \frac{N_1+N_3}{N}\nonumber\\
\!\!\!\!\!\!\!\!\!&&\geq\frac{LN_2+K(K-1)N_3-(K(K-1)+L+1)}{K(K-1)(N_2+N_3)+(K(K-1)-L-1)}
\label{eq:d_i}
\end{eqnarray}
for $i\in\{2,\cdots,K\}$.
Thus, from (\ref{eq:d_1}) and (\ref{eq:d_i}), the achievable DoF is given by
\begin{eqnarray*}
\sup_n \sum_{i=1}^K d_i(n)= \frac{K}{2}+\frac{L}{2(K-1)}.
\end{eqnarray*}
By combining the above three cases, we show that the achievable DoF is $\min\left\{K,\frac{K}{2}+\frac{L}{2(K-1)}\right\}$, which completes the proof.
\end{proof}

\begin{remark} \label{remark:converse}
If $L\geq K(K-1)$, then $d_{\Sigma}=K$. The achievability is given by Theorem \ref{th:K_L_K}.
The converse can be shown straightforwardly from the cut-set bound.
\end{remark}

Let us compare the resulting DoF with that of the decode-and-forward (DF) based relaying of~\cite{Viveck2:09}, which is given by $\frac{KL}{K+L-1}$.
Notice that if $K\leq 5$ than the AF based relaying provides better DoF but if $K>5$ there exist values of $L$ for which the DF based scheme is better.
Specifically, if the number of relays is less than $\frac{1}{2}(K-1)(K-1-\sqrt{K(K-6)+1})$ or greater than $\frac{1}{2}(K-1)(K-1+\sqrt{K(K-6)+1})$, the AF based relaying is better and the DF is better otherwise.

\section{Achievability for $K$-user $K$-hop Networks}
In this section, we consider a $K$-user $K$-hop network having $K$ nodes in each layer and assume $K$ is even.

\begin{figure}[t!]
  \begin{center}
  \scalebox{0.5}{\includegraphics{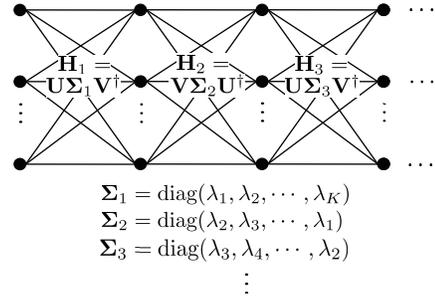}}
  \caption{Opportunistic interference cancellation using relays.}
  \label{fig:opportunistic_int_mit}
  \end{center}
\end{figure}

\subsection{Opportunistic Interference Cancellation}
In this scheme, relays in each layer transmit their received symbols without any modification.
However, they transmit them when the channel matrix of the next hop satisfies a certain condition.
Since the relays may have to wait for such a channel instance for a long time, they need to store some of their past received symbols. To achieve this, we assume block Markov coding is used at relays, i.e., each relay receives a block of $N$ symbols and transmits them (in a permuted order) in the next $N$ time slots. Therefore, there will be one block delay at each hop. To simplify notations, we omit the block indices as before. In the following we only describe how a set of transmitted symbols at any given time at the sources flows through the network.

Let $\mathbf{x}_1$ denote the vector of transmit symbols at any given time at the source nodes. Furthermore, assume it is transmitted through channel matrices $\mathbf{H}_1$ (first hop) through $\mathbf{H}_K$ (last hop).
We define the channel pairing rule from $\mathbf{H}_1$ to $\mathbf{H}_K$ so that the resulting $\mathbf{H}_K\cdots \mathbf{H}_1$ becomes a scaled identity matrix.
Fig. \ref{fig:opportunistic_int_mit} illustrates the basic pairing rule.
Applying the singular value decomposition (SVD), $\mathbf{H}_1$ can be represented as $\mathbf{H}_1=\mathbf{U}\mathbf{\Sigma}_1\mathbf{V}^{\dagger}$, where $\mathbf{U}$ consists of left singular vectors, $\mathbf{\Sigma}_1=\diag\{\lambda_1\,\lambda_2,\cdots,\lambda_K\}$ is the diagonal matrix with ordered singular values, and $\mathbf{V}$ consists of right singular vectors.
Then we choose the next hop channel matrix to be $\mathbf{H}_2=\mathbf{V}\mathbf{\Sigma}_2\mathbf{U}^{\dagger}$, where $\mathbf{\Sigma}_2$ is given as the singular value matrix of $\mathbf{H}_1$ by cyclic shifting the singular values, i.e., $\mathbf{\Sigma}_2=\diag(\lambda_2,\lambda_3,\cdots,\lambda_1)$.
In the same manner, $\mathbf{H}_3$ can be determined from $\mathbf{H}_2$, and so on.
Notice that these paired channels provide $\mathbf{H}_K\cdots \mathbf{H}_1=(\prod_{i=1}^K\lambda_i)\mathbf{I}_K$ meaning $K$ S--D pairs can communicate concurrently without interference.
Similar concepts of opportunistic channel pairing can be found in \cite{Nazer:09, Jafar:09, JeonITA:09} for single-hop networks and in \cite{Jeon:09} for finite-field multi-hop networks.

\subsubsection{Definition of $F_m(\mathbf{H})$}
Let $\mathbf{H}\in\mathbb{C}^{K\times K}$ be a channel instance whose SVD is given by $\mathbf{H}=\mathbf{U}\mathbf{\Sigma}\mathbf{V}^{\dagger}$.
For $m\in\{1,\cdots,K\}$, we define
\begin{equation}
F_m(\mathbf{H})\triangleq\begin{cases}
            \mathbf{U}\mathbf{P}^{m-1}\mathbf{\Sigma}(\mathbf{P}^{m-1})^T\mathbf{V}^{\dagger} & \mbox{ if }m=\mbox{odd} \\
            \mathbf{V}\mathbf{P}^{m-1}\mathbf{\Sigma}(\mathbf{P}^{m-1})^T\mathbf{U}^{\dagger} & \mbox{ if }m=\mbox{even}
         \end{cases}
\label{eq:f_m}
\end{equation}
where
\begin{equation*}
\mathbf{P}=\left[
             \begin{array}{cc}
               \mathbf{0}_{(K-1)\times1}& \mathbf{I}_{K-1} \\
               1 & \mathbf{0}_{1\times(K-1)} \\
             \end{array}
           \right]
\end{equation*}
is the permutation matrix such that the diagonal elements of $\mathbf{P}\mathbf{\Sigma}\mathbf{P}^T$ are equal to the cyclic shift of the diagonal elements of $\mathbf{\Sigma}$.
From the definition, $F_1(\mathbf{H})$ is given by $\mathbf{H}$, where we assume $\mathbf{P}^0=\mathbf{I}_K$.
Note that
\begin{equation*}
\prod_{m=1}^{K}F_m(\mathbf{H})=|\det(\mathbf{H})|\mathbf{I}_K,
\end{equation*}
which is a scaled identity matrix\footnote{ In this paper, $\prod_{m=1}^K \mathbf{A}_m$ denotes $\mathbf{A}_K\mathbf{A}_{K-1}\cdots \mathbf{A}_1$.}.
For notational simplicity, we will use the notation $\Pr(F_m(\mathbf{H}))$ to denote $\Pr(\mathbf{H}_m=F_m(\mathbf{H}))$.
For $m=1$, we will also use $\Pr(\mathbf{H})$ to denote $\Pr(\mathbf{H}_1=\mathbf{H})$ because $F_1(\mathbf{H})=\mathbf{H}$.

\begin{lemma} \label{lemma:Pr_F}
Suppose that the channel coefficients are i.i.d. drawn from a continuous distribution and $\Pr(\mathbf{H}_m)$ is a function of $\|\mathbf{H}_m\|_F$ only.
Then $F_m(\mathbf{H})$ is uniquely determined by $\mathbf{H}$ and
\begin{equation*}
\Pr(\mathbf{H})=\Pr(F_1(\mathbf{H}))=\cdots=\Pr(F_K(\mathbf{H}))
\end{equation*}
for all $\mathbf{H}\in \mathbb{C}^{K\times K}$.
\end{lemma}
\begin{proof}
we refer readers to the full paper \cite{Jeon2:09}.
\end{proof}
\begin{remark} \label{remark:gaussian}
Suppose that the channel coefficients are i.i.d. drawn from $\mathcal{CN}(0,1)$, i.e., Rayleigh fading.
Then $\Pr(\mathbf{H}_m)$ is a function of $\|\mathbf{H}_m\|_F$ only.
\end{remark}

\subsubsection{Transmission scheme}
Let $\mathbf{H}^{\Delta}$ denote the quantized channel matrix in $\Delta(\mathbb{Z}^{K\times K}+j\mathbb{Z}^{K\times K})$ and $\mathcal{H}_1(\mathbf{H}^{\Delta})$ denote the set of all $\mathbf{H}_1$ whose closest point in $\Delta(\mathbb{Z}^{K\times K}+j\mathbb{Z}^{K\times K})$ is equal to $\mathbf{H}^{\Delta}$, respectively.
We further define
\begin{equation}
\mathcal{H}_{m}(\mathbf{H}^{\Delta})=\left\{\mathbf{H}_{m}\big|\mathbf{H}_{m}=F_{m}(\mathbf{H}),\mathbf{H}\in\mathcal{H}_1(\mathbf{H}^{\Delta})\right\}
\end{equation}
for $m\in\{2,\cdots,K\}$.
From Lemma \ref{lemma:Pr_F}, one can easily derive
\begin{equation}
\Pr(\mathcal{H}_{1}(\mathbf{H}^{\Delta}))=\cdots=\Pr(\mathcal{H}_{K}(\mathbf{H}^{\Delta}))\triangleq \Pr(\mathbf{H}^{\Delta})
\label{eq:pr_H_delta}
\end{equation}
for all $\mathbf{H}^{\Delta}\in\Delta(\mathbb{Z}^{K\times K}+j\mathbb{Z}^{K\times K})$.

For transmission, we use the channel $\mathcal{H}_m(\mathbf{H}^{\Delta})$ only when $g_{\min}K\leq\|\mathbf{H}^{\Delta}\|_F\leq g_{\max}K$.
Since $\Pr(h_{ji,m}=h)=0$, the probability of channel utilization can be arbitrarily close to one by setting $g_{\min}$ and $g_{\max}$ arbitrarily small and large, respectively, which does not affect the DoF.

For all $\mathbf{H}^{\Delta}$ satisfying $g_{\min}K\leq\|\mathbf{H}^{\Delta}\|_F\leq g_{\max}K$, the sources transmit their messages to the nodes in the next layer through $\mathbf{H}_1\in\mathcal{H}_1(\mathbf{H}^{\Delta})$ and the relays in the $m$-th layer amplify and forward their received signals to the nodes in the next layer through $\mathbf{H}_m\in\mathcal{H}_m(\mathbf{H}^{\Delta})$, where $m\in\{2,\cdots,K\}$.
That is, we set
\begin{equation}
\mathbf{x}_m=\gamma_m \mathbf{y}_{m-1}.
\label{eq:x_m}
\end{equation}

Suppose that messages are transmitted through a series of particular channel matrices $\mathbf{H}_1$ to $\mathbf{H}_K$ such that $\mathbf{H}_m\in \mathcal{H}_m(\mathbf{H}^{\Delta})$.
Then from (\ref{eq:input_output_vec}) and (\ref{eq:x_m}), we obtain
\begin{equation*}
\mathbf{y}_K=\left(\prod_{m=2}^K\gamma_{m}\right)\left(\prod_{m=1}^K \mathbf{H}_m\right)\mathbf{x}_1+\mathbf{z}_{AF}+\mathbf{z}_K,
\end{equation*}
where
\begin{equation}
\mathbf{z}_{AF}=\sum_{i=2}^K\left(\prod_{j=i}^K\gamma_{j}\right)\left(\prod_{j=i}^K \mathbf{H}_j\right)\mathbf{z}_{i-1}
\label{eq:n_af}
\end{equation}
denotes the accumulated noise due to AF relaying.
Let $\mathbf{H}_1=\mathbf{H}$ and $\mathbf{H}_m=F_m(\mathbf{H})+\mathbf{\Delta}_m$ for $m\in\{1,\cdots,K\}$, where $\mathbf{\Delta}_m$ is the quantization error matrix of $\mathbf{H}_m$ with respect to $F_m(\mathbf{H})$.
From the definition of $F_m(\mathbf{H})$, we know that $\mathbf{\Delta}_1=\mathbf{0}_{K\times K}$.
Then we obtain
\begin{eqnarray*}
\mathbf{y}_K=\left(\prod_{m=2}^K\gamma_{m}\right)|\det(\mathbf{H})|\mathbf{x}_1+\mathbf{\Delta}_{tot}\mathbf{x}_1+\mathbf{z}_{AF}+\mathbf{z}_K,
\end{eqnarray*}
where
\begin{eqnarray}
\mathbf{\Delta}_{tot}\!\!\!\!\!\!\!\!\!&&=\sum_{i=1}^K\mathbf{\Delta}_i\prod_{j=1,j\neq i}^K F_j(\mathbf{H})\nonumber\\
&&+\sum_{i=1}^K\sum_{j<i}^K\mathbf{\Delta}_i\mathbf{\Delta}_j\prod_{k=1,k\neq i,j}^K F_k(\mathbf{H})\nonumber\\
&&+\cdots+\prod_{i=1}^K \mathbf{\Delta}_i,
\label{eq:delta_tot}
\end{eqnarray}
which is the total quantization error matrix.
Then the signal to interference and noise ratio (SINR) of the $k$-th destination is lower bounded by
\begin{equation}
\mbox{SINR}^{\Delta}_k\geq\frac{\left(\big|\prod_{m=2}^K\gamma_{m}\big||\det(\mathbf{H})|-\|\mathbf{\Delta}_{tot}\|_F\right)^2P}{1+\|\mathbf{\Delta}_{tot}\|^2_F KP+E\left(\|\mathbf{z}_{AF}\|^2_F\right)}.
\label{eq:sinr}
\end{equation}

The following two lemmas show achievable rates when the quantization interval $\Delta$ tends to zero.

\begin{lemma} \label{lemma:delta_tot}
As $\Delta\to 0$, $\|\mathbf{\Delta}_{tot}\|_F$ converges to zero.
\end{lemma}
\begin{proof}
we refer readers to the full paper \cite{Jeon2:09}.
\end{proof}

\begin{lemma} \label{lemma:closed_form}
As $\Delta\to 0$, the following rate is achievable:
\begin{eqnarray*}
R_k=\int_{g_{\min}K\leq\|\mathbf{H}\|_F\leq g_{\max}K}\!\!\!\!\!\log\left(1+\mbox{SINR}_k\right)\Pr(\mathbf{H})d\mathbf{H}-\epsilon_n,
\end{eqnarray*}
where
\begin{equation}
\mbox{SINR}_k=1+\frac{\prod_{m=2}^K\gamma^2_{m}\det(\mathbf{H})^2P}{1+E\left(\|\mathbf{z}_{AF}\|^2_F\right)}
\label{eq:closed_form}
\end{equation}
and $\epsilon_n>0$ converges to zero as $n$ tends to infinity.
\end{lemma}
\begin{proof}
we refer readers to the full paper \cite{Jeon2:09}.
\end{proof}

\subsection{DoF of $K$-user $K$-hop networks}
Based on the previous lemmas, we derive the achievable DoF of the $K$-user $K$-hop network.
\begin{theorem} \label{th:K_user_K_hop}
Suppose a $K$-user $K$-hop network with $K$ nodes in each layer.
If $K$ is even and $\Pr(\mathbf{H})$ is a function of $\|\mathbf{H}\|_F$ only, then $d_{\Sigma}=K$.
\end{theorem}
\begin{proof}
Because we choose channel matrices satisfying $g_{\min}K\leq\|\mathbf{H}_m\|_F\leq g_{\max}K$ for transmission, the relay coefficients satisfying power constraint $P$ can be bounded between strictly positive finite minimum and maximum values.
Hence the terms $\prod_{m=2}^K\gamma^2_{m}$ and $E\left(\|\mathbf{z}_{AF}\|_F^2\right)$ in (\ref{eq:closed_form}) does not affect the DoF.

Let us now consider $\det(\mathbf{H})$ in (\ref{eq:closed_form}) that can be represented as
\begin{equation*}
\det(\mathbf{H})=\sum_{j=1}^K h_{ij}C_{ij},
\end{equation*}
where $h_{ij}$ denotes $(i,j)$-th element of $\mathbf{H}$ and $C_{ij}$ is the cofactor, which is a function of $\cup_{k\neq i,l\neq j}\{h_{kl}\}$.
Hence, for given $C_{i1}$ to $C_{iK}$ and $h_{i2}$ to $h_{ik}$, $\det(\mathbf{H})$ becomes zero if and only if
\begin{equation*}
h_{i1}=-\frac{h_{i2}C_{i2}+\cdots+h_{iK}C_{iK}}{C_{i1}}.
\end{equation*}
However this event occurs with probability $0$. Therefore, this does not affect the DoF and
we can show DoF of $K$ is achievable. For detailed proof, we refer readers to the full version of this paper.
The converse can be shown similarly as in Remark \ref{remark:converse}, which completes the proof.
\end{proof}

Notice that, from Remark \ref{remark:gaussian}, the network with i.i.d. Gaussian channel distributions is a special class of Theorem \ref{th:K_user_K_hop}.


\end{document}